\documentclass[journal,english,twocolumn,10pt,letterpaper]{IEEEtran}
\usepackage{babel}
\pdfoutput=1
\usepackage[utf8x]{inputenc}
\usepackage{graphicx}
\usepackage{subfigure}
\usepackage{amsmath}
\usepackage{amsthm}
\usepackage{amssymb}
\usepackage[svgnames]{xcolor}
\usepackage{balance}
\usepackage{url}
\usepackage{subfigure}
\usepackage{breqn}
\usepackage{listings}
\usepackage[colorlinks,bookmarksopen,bookmarksnumbered,linkcolor=darkgray,citecolor=darkgray,urlcolor=darkgray]{hyperref}
\usepackage{tikz}

\usepackage{changebar}

\newcommand{\added}[1]{{#1}}
\newcommand{\addedfragile}[1]{{#1}}

\newtheorem{theorem}{Theorem}
\newtheorem{proposition}[theorem]{Proposition}

\newcommand{\Constantconstant}[2]{\mathrm{#1}_{\mathrm{#2}}}
\newcommand{\Ts}{\Constantconstant T s}
\newcommand{\Tw}{\Constantconstant T w}
\newcommand{\Qw}{\Constantconstant Q w}
\newcommand{\Te}{\Constantconstant T e}
\newcommand{\Tmax}{\Constantconstant T \max}
\newcommand{\Toff}{\mathrm{T}_{\mathrm{off}}}

\lstdefinelanguage{JavaScript}{
  keywords={typeof, new, true, false, catch, function, return, null, catch, switch, var, if, in, while, do, else, case, break},
  keywordstyle=\bfseries,
  ndkeywords={class, export, boolean, throw, implements, import, this},
  ndkeywordstyle=\color{darkgray}\bfseries,
  identifierstyle=\color{black},
  sensitive=false,
  comment=[l]{//},
  morecomment=[s]{/*}{*/},
  commentstyle=\ttfamily,
  stringstyle=\ttfamily,
  morestring=[b]',
  morestring=[b]"
}

\newcommand\copyrighttext{%
  \footnotesize \textcopyright 2015 IEEE. Personal use of this material is permitted.
  Permission from IEEE must be obtained for all other uses, in any current or future 
  media, including reprinting/republishing this material for advertising or promotional 
  purposes, creating new collective works, for resale or redistribution to servers or 
  lists, or reuse of any copyrighted component of this work in other works. 
  DOI: \href{https://dx.doi.org/10.1109/JSYST.2015.2466086}{10.1109/JSYST.2015.2466086}
}
\newcommand\copyrightnotice{%
\begin{tikzpicture}[remember picture,overlay]
\node[anchor=south,yshift=10pt] at (current page.south) {\fbox{\parbox{\dimexpr\textwidth-\fboxsep-\fboxrule\relax}{\copyrighttext}}};
\end{tikzpicture}%
}

\begin{document}

\title{Optimum Traffic Allocation in Bundled Energy Efficient Ethernet Links}

\author{Miguel~Rodríguez-Pérez,~\IEEEmembership{Member,~IEEE}, Manuel
  Fernández-Veiga,~\IEEEmembership{Senior Member,~IEEE,} \\
  Sergio~Herrería-Alonso, \added{Mariem} Hmila and~Cándido~López-García%
  \thanks{The authors are with the Telematics Engineering Dept., Univ.~of
    Vigo, 36$\,$310~Vigo, Spain. Tel.:~+34$\,$986$\,$813459;
    fax:~+34$\,$986$\,$812116; email:~\protect\url{miguel@det.uvigo.es} (M.
    Rodríguez-Pérez).}
  \thanks{%
    \addedfragile{Work supported by the European Regional Development Fund (ERDF) and the
      Galician Regional Government under agreement for funding the Atlantic Research
      Center for Information and Communication Technologies
      (\href{http://atlanttic.uvigo.es/en/}{AtlantTIC}).}%
  }
}

 
\maketitle
\copyrightnotice
\begin{abstract}
  The energy demands of Ethernet links have been an active focus of research
  in the recent years. This work has enabled a new generation of Energy
  Efficient Ethernet (EEE) interfaces able to adapt their power consumption to
  the actual traffic demands, thus yielding significant energy savings. With
  the energy consumption of single network connections being a solved problem,
  in this paper we focus on the energy demands of link aggregates that are
  commonly used to increase the capacity of a network connection. We build on
  known energy models of single EEE links to derive the energy demands of the
  whole aggregate as a function on how the traffic load is spread among its
  powered links. We then provide a practical method to share the load that
  minimizes overall energy consumption with controlled packet delay, and prove
  that it is valid for a wide range of EEE links. Finally, we validate our
  method with both synthetic and real traffic traces captured in Internet
  backbones.
\end{abstract}
 

\section{Introduction}
\label{sec:introduction}

\IEEEPARstart{E}{nergy consumption} is nowadays a global source of concern for
both economic and environmental reasons. Networking equipment alone consumes
1.8\% of the world's electricity, and that number is currently increasing at a
10\% rate annually~\cite{lambert12:_world}. If we \added{just focus} on data
centers, between 15 and 20\% of electricity is used for
networking~\cite{heller10:_elast}. These reasons are spurring the development
of more power efficient networking equipment.

A direct result of these efforts is the IEEE~802.3az standard~\cite{802.3az}
which provides a new \emph{idle mode} for Ethernet physical interfaces. This
new mode only needs a small fraction of the power used in normal operation,
but no traffic can be transmitted nor received while the interface stays in
the idle mode. Since there is an implicit trade-off between energy consumption
and frame delay, these new Energy Efficient Ethernet (EEE) interfaces need a
governor that decides when to enter and exit this idle mode. In fact, several
alternatives have already been proposed in the
literature~\cite{gupta07:_using_low_power_modes_for,rodriguez09:_improv_oppor_sleep_algor_lan_switc,reviriego10:_burst_tx_eee,christensen10:_the_road_to_eee}
and have been later validated by both
empirical~\cite{reviriego09:_perf_eval_eee,reviriego11:_initial_evaluat_energ_effic_ether}
and analytic
means~\cite{herreria11:_oppor_ether,herreria11:_power_savin_model_burst_trans,marsan11:_simpl_analy_model_energ_effic_ether,herreria12:_gi_g_model_gb_energ_effic_ether,aksic14:_packet_ether}.
These works have provided us with the tools needed to accurately estimate the
power savings of EEE for any arrival traffic pattern with the more prevalent
idle mode governors and to properly tune them to maximize energy savings.

With the energy consumption problem of single Ethernet links mostly solved we
focus in this paper on the power demands of network connections formed by
multiple EEE links, either by link aggregation~\cite{ieee_std_802.1ax} or some
other proprietary means. Despite the existence of EEE for saving energy in the
individual components of the bundle, the global consumption of an aggregate
may be severely affected by how the incoming traffic is shared among its
powered up links. \added{In fact, the power profiles of the individual EEE links are
not linear, as their energy demands do not grow proportionally to the offered
load. This makes the overall power consumption dependent on the actual traffic
share among the links of the aggregate.
}

\added{The main goal of this paper is to obtain the optimum share of traffic among
  the links of an aggregate from an energy efficiency perspective. As far as
  we know, this is the first paper to tackle this issue. We propose a
  \emph{water-filling} algorithm, where traffic is only transmitted on a given
link if all the previous ones are already being used at their maximum capacity
and show that it is optimum for various relevant traffic arrival patterns.
Additionally, we also propose a practical implementation of the algorithm that
can be applied with minimal computational needs in the firmware of
Ethernet line cards.}

The rest of this paper is organized as follows. \added{We introduce some work
  related to Energy Efficient Ethernet in Section~\ref{sec:related-work}.}
Section~\ref{sec:problem-description} provides a formal description of the
problem at hand. Section~\ref{sec:toff} analyzes the concavity of the cost
function of the main EEE algorithms. Section~\ref{sec:delay-control} details a
practical algorithm to implement water-filling. The results are commented in
Section~\ref{sec:results}. Finally, Section~\ref{sec:conclusions} ends the
paper with our conclusions.

\added{
  \section{Related Work}
  \label{sec:related-work}
  
  There are several areas where energy can be saved in the current Internet
  that were first identified in~\cite{gupta03:_green_of_inter}. The existence
  of spare installed capacity was one of the identified aspects. Several works
  proposed to power off unused links during low load periods concentrating
  traffic on just a few network
  paths~\cite{chiaraviglio12:_minim_isp_networ_energ_cost,addis14:_energ_manag_throug_optim_routin,garroppo11:_energ_aware_routin_based_energ_charac_devic,garroppo13:_does_traff_consol_alway_lead,galan2013:_using}.
  Of all these
  proposals,~\cite{garroppo11:_energ_aware_routin_based_energ_charac_devic,garroppo13:_does_traff_consol_alway_lead}
  also take into consideration aggregated links between two network devices.
  However, all these works focus on long timescales, usually hours, while we
  are interested in much lower timescales, as such, both approaches can be
  seen as complementary. Links (and network paths) can be powered off when the
  long-term traffic load is low enough, while, for the short timescales,
  another approach should be used to reduce the energy usage of those links in
  the aggregate that remain active.

  Another source of inefficiency identified in~\cite{gupta03:_green_of_inter}
  was the physical interfaces of network devices. At that time, physical
  interfaces drew a constant amount of power, regardless of the actual traffic
  load. Preliminary works tried to mitigate this either by adapting the
  transmission speed~\cite{gunaratne08:_reduc_energ_consum_ether_adapt}, with
  lower speeds demanding less power, or by briefly switching off the physical
  interfaces when there is none or very little traffic to
  send~\cite{gupta07:_using_low_power_modes_for,rodriguez09:_improv_oppor_sleep_algor_lan_switc}.
  Finally, the IEEE~802.3az~\cite{802.3az} standard was sanctioned providing a
  new low power mode to physical Ethernet interfaces that could be used when
  there was no need to send traffic. 

  New research then focused on the best way to use this new low power mode.
  The straightforward solution, entering low power mode as soon as all traffic
  has been transmitted, and returning to the normal mode with the first packet
  arrival, called \emph{frame transmission}, was experimentally studied
  in~\cite{reviriego11:_initial_evaluat_energ_effic_ether}. A first analytic
  study appeared in~\cite{larrabeiti11:_towar_gb_ether} for Poisson traffic,
  while another analysis considering arrivals of packets trains to take into
  account burst traffic arrivals was presented
  in~\cite{marsan11:_simpl_analy_model_energ_effic_ether}.

  Another explored possibility to make use of the new power mode consists on
  waiting for the arrival of several packets before returning to active
  mode~\cite{reviriego10:_burst_tx_eee,aksic14:_packet_ether}. This mode,
  known as \emph{packet coalescing} or \emph{burst transmission}, avoids
  unnecessary transitions between the normal and low power modes greatly
  improving the energy savings at the cost of additional delays. There exists
  analytic models for the power savings of burst transmission for both Poisson
  traffic~\cite{herreria11:_power_savin_model_burst_trans} and for its
  delay~\cite{Bolla201416,herreria12:_optim_energ_effic_ether}. A general
  model for general arrival patterns for both frame and burst transmission
  covering both power usage and delay can be found
  in~\cite{herreria12:_gi_g_model_gb_energ_effic_ether}.
 
  New research tries to find innovative ways to govern the use of the low
  power mode, see for instance~\cite{cenedese14:_ether} that exploits traffic
  self-similarity to obtain the best duration of the low power interval, in
  such a way that maximizes energy savings for a given maximum allowable
  additional delay.

}

\section{Problem Description}
\label{sec:problem-description}

In transmission networks, it is customary to bundle several homogeneous links,
i.e., links with similar transmission technology, as a cheap way for scaling
up the aggregate transmission rate between two endpoints. The bundle can be
seen and managed either as a set of independent links or as a unit by the
traffic management algorithms and the upper layer protocols. In the latter
case, the traffic is split among the individual links in the bundle
considering the optimization of a given performance metric. We focus in this
paper on the optimum allocation of traffic when the bundle components are
Energy Efficient Ethernet (EEE) links (IEEE 802.3az~\cite{802.3az}), from the
point of view of total energy consumption minimization. The profile of energy
consumption in EEE links has been analyzed in many
works~\cite{reviriego09:_perf_eval_eee,reviriego11:_initial_evaluat_energ_effic_ether,herreria11:_oppor_ether,herreria11:_power_savin_model_burst_trans,marsan11:_simpl_analy_model_energ_effic_ether,herreria12:_gi_g_model_gb_energ_effic_ether,Vitturi2015228,Bolla201416},
and has been shown to be highly sensitive to the statistical variability of
the incoming traffic. Thus, further gains in energy efficiency may be realized
if the total traffic load offered to the bundle is properly allocated to
individual links.

We consider a bundle comprising $N$ identical transmission links. The traffic
demand to the bundle is $X$, and $E(x_i)$ is the energy consumption of link
$i = 1, \dots, N$, where $x_i$ stands for the traffic rate in that link. Link
capacities are denoted by $C_i$, for $i = 1, \dots, N$.. 

Our goal is to minimize the overall consumption of the bundle $E\added{_B}(x_1, \dots,
x_N) = \sum_{i = 1}^N E(x_i)$, that is
\begin{equation}
  \label{eq:problem}
  \min \,E\added{_B}(x_1, \dots, x_N)
\end{equation}
such that
\begin{align}
  C \geq x_i &\geq 0, \text{ and }\\
  \sum_{i = 1}^N x_i &= X
\end{align}
where 
\begin{equation}
  \label{eq:link-cost}
  E(x_i) = 1 - (1 - \sigma_{\text{off}})(1 - \rho_i)
  \frac{\mathrm{T}_{\mathrm{Toff}}(\rho_i)}{\mathrm{T}_{\mathrm{Toff}}(\rho_i) + \mathrm{T}_{\mathrm{s}} + \mathrm{T}_{\mathrm{w}}}
\end{equation}
is the normalized energy consumption of link $i$, as shown
in~\cite{herreria12:_gi_g_model_gb_energ_effic_ether}.
In~\eqref{eq:link-cost}, $\Ts$ and $\Tw$ are constant and account for,
respectively, the transition times needed to enter and exit the idle mode
defined in IEEE~802.\added{3}az. $\Toff(\rho_i)$ is the average time spent by the
interface in the idle state for a given input load. Note that $\Toff(\rho_i)$
depends on both the actual traffic arrival pattern and the idle state
governor. Finally, $\sigma_{\text{off}}$ is simply the fraction of energy
consumed by the interface in the idle state compared to its energy consumption
in the active state and $\rho_i = x_i / C$ is the normalized traffic load on
the \added{link. So}, \eqref{eq:problem} is a standard minimization problem
amenable to analysis provided that $E\added{_B}(x_1, \dots, x_N)$ is a well-behaved
function.

\subsection{Optimum allocation}
\label{sec:optimal-allocation}

In this Subsection, we prove that for certain functions $\Toff(\cdot)$ the
solution to the optimum allocation is a simple sequential water-filling
algorithm: each link capacity is fully used before sending traffic through a
new, idle link. Clearly, \eqref{eq:problem} is a concave separable
optimization problem when the objective function is concave and we have the
following simple result.
\begin{proposition}
\label{propo:waterfill}
\added{ If $E(\cdot)$ is a strictly concave function and
  $\{C_1, C_2, \ldots, C_N\}$ with $C_i \ge C_{i+1}$ are the link capacities,
  then 
$E\added{_B}(x_1,\ldots, x_N)$ is minimum if $x_i = \min\{C_i,
X-\sum_{j<i}x_j\}$.
}
\end{proposition}
\begin{proof}
  The proof is a direct consequence of the subadditivity of $E(\cdot)$ and is
  given in \added{Appendix}~\ref{sec:proof}.
\end{proof}

Now we derive sufficient conditions for the concavity of the cost function
$E(\cdot)$. Recall from~\eqref{eq:link-cost} that $E(\cdot)$ depends on some
constants related to the interface hardware and the statistical variability of
the incoming traffic. We will try to understand what conditions \added{must
  satisfy $\Toff$}, which is the only traffic-dependent term. For clarity and
simplicity, in the following we use the notations $f(\rho) = \Toff(\rho)$ and
$t(\rho) = E(\rho)$. We will further assume that $f(\rho)$ is
decreasing\footnote{$f(\rho)$ computes the average time spent by the interface
  in the idle state, so it is reasonable to assume it is decreasing when the
  traffic load is higher.} and continuously differentiable in
$\rho \in (0, 1)$.
\begin{proposition}
  \label{prop:convex}
  Let $f(x)$ be a function $f: [0, 1] \to \mathbb{R}^{+}$, decreasing and with
  continuous derivatives. Let $a, b > 0$ and consider the function
  \begin{equation}
    t(x) = 1 - a(1- x) \frac{f(x)}{f(x) + b}.
  \end{equation}
  Under these definitions, $t(x)$ is concave if
  \begin{equation}
    \label{eq:convex-condition}
    f^{\prime\prime}(x) \bigl( f(x) + b \bigr) \geq 2 \bigl( f^\prime(x) \bigr)^2.
  \end{equation}
\end{proposition}
\begin{proof}
  The proof is provided in \added{Appendix}~\ref{sec:proof-proposition-convex}.
\end{proof}

Proposition~\ref{prop:convex} applies trivially to the function $E(\cdot)$
setting $a = 1 - \sigma_{\text{off}}$ and $b = \Ts + \Tw$, \added{so} we have
derived a simple sufficient condition for the $\Toff(\cdot)$ term that makes
$E(\cdot)$ concave and the optimization problem easily solvable.

\section{Analysis of Frame and Burst Transmission}
\label{sec:toff}

In this Section we check whether the known formulas for the average sleeping
time in EEE satisfy the condition of Proposition~\ref{prop:convex}. According
to~\cite{herreria12:_gi_g_model_gb_energ_effic_ether} the time $\Toff(\cdot)$
depends both on the incoming traffic characteristics and the threshold
algorithm used to switch between the idle and the active states in the
Ethernet interface. There are two main approaches, the \emph{frame
  transmission} algorithm and the \emph{burst transmission} one, that we
consider next.

\subsection{Frame Transmission}

Frame transmission is a straightforward use of the idle mode. Under frame
transmission, the physical interface is put in idle mode as soon as the last
frame in the queue has been transmitted, and normal operation is restored as
soon as new traffic arrives at the networking interface. For many common
traffic patterns this operating mode does not produce great energy savings, as
there is a transition period every time the interface changes its operating
mode that draws some energy.
From~\cite{herreria12:_gi_g_model_gb_energ_effic_ether}, for the frame
transmission algorithm
\begin{equation}
  \label{eq:frame-t}
  \Toff^{\text{frame}}(\rho) = \int_{\Ts}^\infty (t - \Ts)
  f_{\rho,\Te}(t) \textrm{d}t,
\end{equation}
where $f_{\rho,\Te}(t)$ denotes the probability density function for traffic
load $\rho$ of the empty period, i.e., the time elapsed since the queue
empties until the subsequent first arrival. When $f_{\rho,\Te}(t)$ is unknown,
\added{from~\cite{herreria12:_gi_g_model_gb_energ_effic_ether}}
equation~\eqref{eq:frame-t} can be approximated by
\begin{equation}
  \label{eq:frame-t-approx}
   \Toff^{\text{frame}}(\rho) \approx \left( \frac{1}{\mu \rho} -
     \Ts \right)^+,
\end{equation}
with $\mu^{-1}$ the average packet transmission duration. Closed formulas
exist when the arrival process follows a Poisson or a deterministic
distribution. In particular, for Poisson arrivals, we have
\begin{equation}
  \Toff^{\text{frame}}(\rho) = \frac{\textrm{e}^{-\mu \rho
      \Ts}}{\mu \rho}.
\end{equation}

\subsubsection{Poisson traffic}
\label{sec:poisson-traffic}

For proving the concavity under the assumption of Poisson arrivals, we start
by noting that $f(\rho) = \textrm{e}^{-\mu\rho \Ts}/ (\mu\rho)$ and substitute
this in~\eqref{eq:convex-condition} with $b = \Ts + \Tw$. The result is the
condition
\begin{equation}
  \begin{split}
    \left( \frac{\mu \Ts^2 \textrm{e}^{-\mu\rho - \Ts}}{\rho} +
      \frac{2 \textrm{e}^{-\mu\rho - \Ts}}{\mu\rho^3} + \frac{2
        \Ts \textrm{e}^{-\mu\rho \Ts}}{\rho^2} \right) \cdot \\
    \left( \frac{\textrm{e}^{-\mu\rho \Ts}}{\mu\rho} + \Ts + \Tw
    \right) > -2 \left( \frac{\textrm{e}^{-\mu\rho \Ts}}{\mu\rho^2} +
      \frac{\Ts \mathrm{e}^{-\mu\rho \Ts}}{\rho} \right)^2,
  \end{split}
\end{equation}
and after some routine simplifications this reduces to
\begin{equation}
  (\Ts + \Tw)
  \mathrm{e}^{\mu\rho \Ts} \bigl(2 + \mu \rho \Ts (2 + \mu\rho
  \Ts) \bigr) > \Ts (2 + \mu\rho \Ts).
\end{equation}
But $\mu\rho \Ts > 0$ and $\mathrm{e}^{\mu\rho \Ts} > 1$, so
\begin{equation}
  \begin{split}
    (\Ts + \Tw) \mathrm{e}^{\mu\rho \Ts} \bigl(2 + \mu \rho
    \Ts (2 + \mu\rho \Ts) \bigr) > \\
    \Ts \mathrm{e}^{\mu\rho\Ts} \bigl(2 + \mu \rho \Ts (2 + \mu\rho \Ts)
    \bigr) > \Ts (2 + \mu\rho \Ts),
  \end{split}
\end{equation}
and~\eqref{eq:convex-condition} is satisfied.

Note that it is important to ascertain that the link consumption function
$E(\cdot)$ is concave for Poisson traffic since, notwithstanding that
Poissonian models are not generally suitable, they are reasonably valid for
real traffic in sub-second timescales~\cite{Karagiannis04} and also for
aggregated traffic in the Internet core~\cite{Vishwanath09}. In any case, in
Section~\ref{sec:results} we test the validity of our assumptions with both
synthetic and real traffic traces collected in Internet links.

Figure~\ref{fig:numerical} shows, for purposes of illustration, a contour plot
of $h^{\prime\prime}(\rho)$ for the function
$h(\rho) = \frac{\Toff^{\text{frame}}(\rho)}{(\Toff^{\text{frame}}(\rho) + \Ts
  + \Tw)}$.
The traffic is Poissonian and the Ethernet link runs at $10$ Gb/s
($b = \Ts + \Tw = 2.28\text{ $\mu$s} + 4.48\text{ $\mu$s}$, as mandated by the
IEEE 802.3az standard~\cite{802.3az}), with packet sizes between $64$ and
$9000$ bytes.\added{\footnote{\addedfragile{Although the maximum Ethernet
      capacity is limited to 1500 bytes, we have tested greater packet sizes
      to account for so called Jumbo-frames.}}} It can be seen that
$h^{\prime\prime}(\rho) \ge 0$ in the region of interest, thus $h(\rho)$ is
convex and
$E(x) = 1 - (1-\sigma_{\mathrm{off}})(1-\rho)h(\rho),\, 0 \le \rho,
\sigma_{\mathrm{off}} < 1$ is concave.

\begin{figure}
  \centering
  \includegraphics[width=0.5\textwidth]{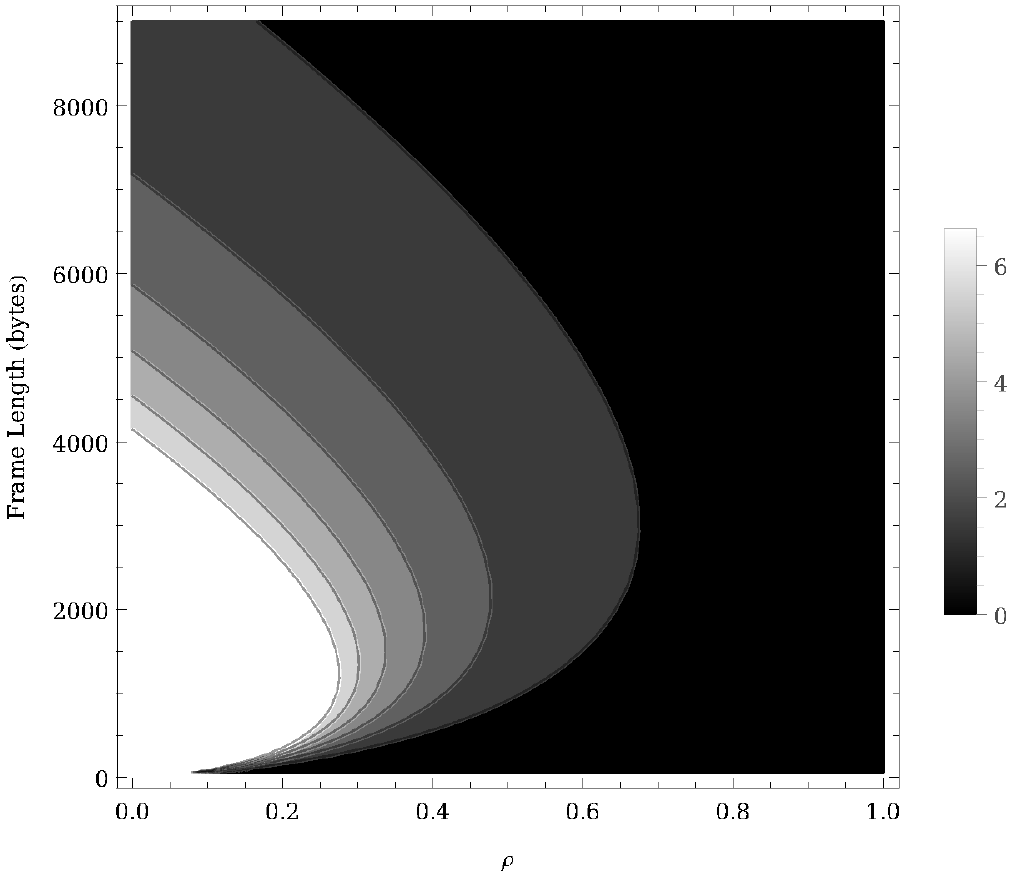}
  \caption{\label{fig:numerical} Contour plot of $h^{\prime\prime}(\rho)$ when
    $h(\rho) =
    \frac{\Toff^{\text{frame}}(\rho)}{\Toff^{\text{frame}}(\rho) +
      \Ts + \Tw}$ for a Poisson arrival process under the frame
    transmission energy-saving algorithm.}
\end{figure}

\subsubsection{General traffic distributions}
\label{sec:gener-traff-distr}

For unknown traffic distributions we must resort to the approximation given
by~\eqref{eq:frame-t-approx}, so we let $f(\rho) = 1/(\mu \rho) - \Ts$ and
$b = \Ts + \Tw$. Now we can immediately substitute in
$f^{\prime\prime}(\rho) (f(\rho) + b) > 2( f^\prime(\rho))^2$ and get
\begin{equation}
  \frac{2 \bigl( \frac{1}{\mu\rho} + \Tw \bigr)}{\mu \rho^3} >
  \frac{2}{\mu^2 \rho^4}.
\end{equation}
After some straightforward cancellations, this is 
\begin{equation}
  \frac{2 \Tw}{\mu \rho^3} > 0,
\end{equation}
which is obviously true.

\subsection{Burst Transmission}

Burst transmission is a simple modification of frame transmission that waits
until a given number of packets $\Qw$ arrive at the network interface before
exiting idle mode. To avoid excessive delays, there is a tunable parameter
$\Tmax$ that limits the wait for the $\Qw$-th frame since the first frame
arrives. The analysis of the burst transmission algorithm is more involved,
for the reason that there is not one but two operating regimes depending on
the traffic load. Fortunately,
\added{\cite{herreria12:_gi_g_model_gb_energ_effic_ether} shows that} the two
operating regimes (low and high traffic load, respectively) can be neatly
separated by the approximate traffic threshold
\begin{equation}
  \rho^\ast \approx \frac{\Qw  - 1}{\mu \Tmax},
\end{equation}
where $\Qw$ and $\Tmax$ are the tunable parameters in the burst transmission
algorithm~\cite{herreria11:_power_savin_model_burst_trans}. As in the previous
Section, we will proceed and check whether, with burst transmission, the link
energy consumption function is concave.

\subsubsection{Low load regime, $\rho < \rho^\ast$}
\label{sec:low-load-regime}

When the traffic load is low, the interface exits the low power mode before a
backlog of $\Qw$ packets accumulates at the queue due to the timer expiry
after waiting for $\Tmax$ seconds. The exact expression for the expected
sojourn time in the low-power state is
(see~\cite{herreria12:_gi_g_model_gb_energ_effic_ether})
\begin{equation}
  \label{eq:burst-t}
  \Toff^{\text{burst, low}}(\rho) = \int_0^\infty (t + \Tmax -
  \Ts) f_{\rho,\Te}(t) \textrm{d}t.
\end{equation}
When $f_{\rho,T_{\text{e}}}(t)$ is unknown, \added{according
  to~\cite{herreria12:_gi_g_model_gb_energ_effic_ether}},~\eqref{eq:burst-t}
can be approximated by
\begin{equation}
  \label{eq:burst-t-approx}
   \Toff^{\text{burst, low}}(\rho) \approx \frac{1}{\mu \rho} +
   \Tmax - \Ts.
\end{equation}
As in the frame transmission algorithm, there exist closed expressions for $
\Toff^{\text{burst, low}}(\rho)$ for some distributions, and
remarkably~\eqref{eq:burst-t-approx} is exact with Poissonian arrivals.

Proving the concavity of $E(\rho)$ in this case is direct. First, note that
$f(\rho) = \Toff^{\text{burst, low}}(\rho) = \Toff^{\text{frame}}(\rho) +
\Tmax$, so that the derivatives $f^\prime$ and $f^{\prime\prime}$ are the same
as in the frame transmission case, and hence
plugging~\eqref{eq:burst-t-approx} into the condition $f^{\prime\prime}(x)
(f(x) + b) > 2 (f^\prime(x))^2$ one can easily check that the inequality
holds.

\begin{figure*}
  \centering
  \subfigure[Low load regime]{
    \includegraphics[width=\columnwidth]{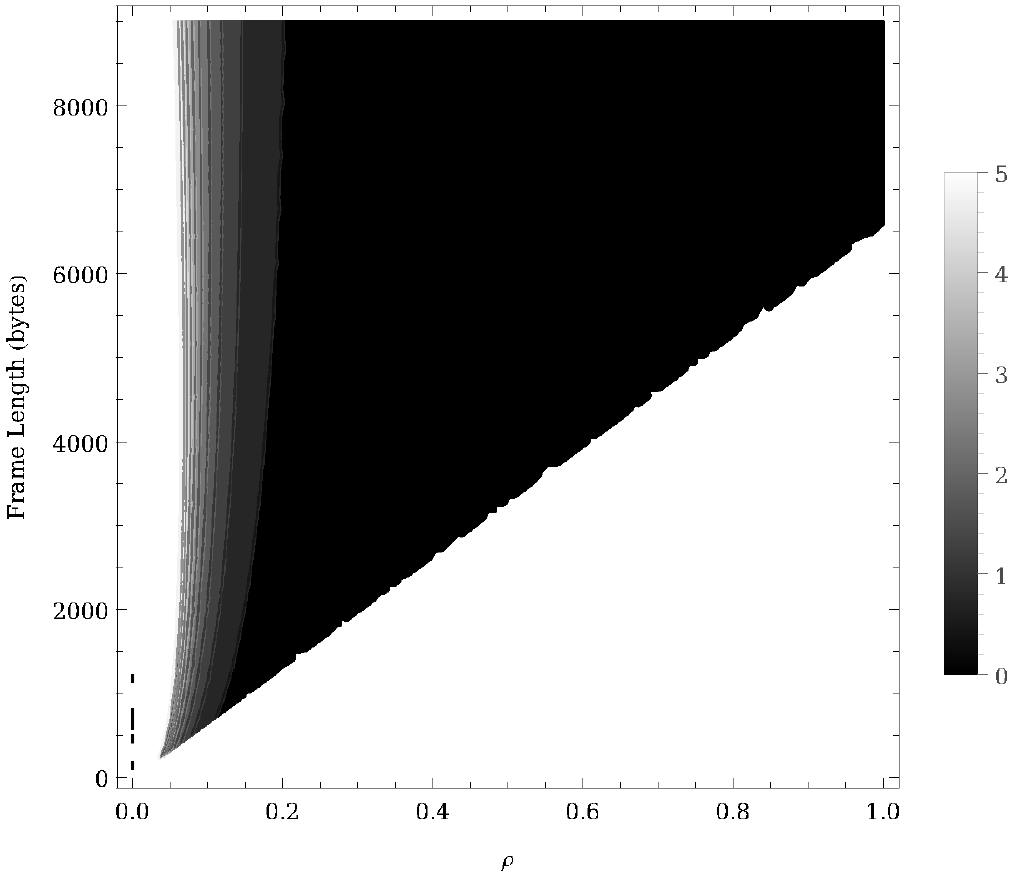}}
  \subfigure[High-load regime]{
    \includegraphics[width=\columnwidth]{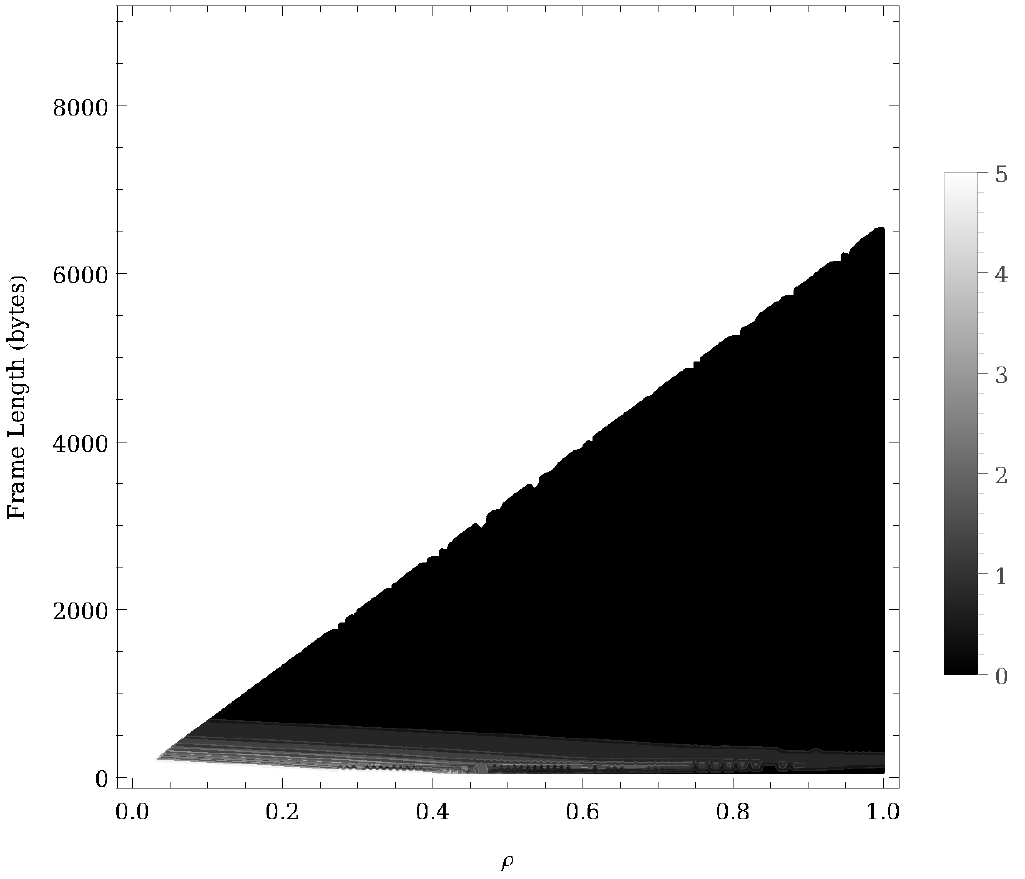}}
  \caption{\label{fig:contour} Contour plots of $h^{\prime\prime}(\rho)$ when
    $h(\rho) = \frac{\Toff^{\text{burst}}(\rho)}{\Toff^{\text{burst}}(\rho) +
      \Ts + \Tw}$ for a Poisson arrival process under the burst transmission
    energy-saving algorithm.}
\end{figure*}

\subsubsection{High load regime, $\rho > \rho^\ast$}
\label{sec:high-load-regime}

When the traffic load is high, the packet burst is much more likely to reach
its maximum size $\Qw$ before the timer expires. Now, the expected sojourn
time in the low-power state is given by
\begin{equation}
  \Toff^{\text{burst,high}}(\rho) = \int_{\Ts}^\infty (t -
  \Ts) f_{\rho,\Qw}(t) \textrm{d}t,
\end{equation}
where, as usual, $f_{\rho,\Qw}(t)$ is the probability density function of the
$\Qw$-th frame arrival epoch after the interface has entered the idle mode.
When the density is unknown, \added{according
  to~\cite{herreria12:_gi_g_model_gb_energ_effic_ether}} the expected time can
be well approximated by
\begin{equation}
  \label{eq:burst-h-approx}
  \Toff^{\text{burst,high}}(\rho) \approx \frac{\Qw}{\mu\rho} -
  \Ts,
\end{equation}
whereas the exact formula for the case of Poissonian arrivals is
\begin{equation}
  \label{eq:burst-h-poisson}
  \Toff^{\text{burst,high}}(\rho) = \frac{\Gamma(\Qw + 1, \mu\rho
    \Ts) - \mu\rho \Ts \Gamma(\Qw,\mu\rho
    T_s)}{\mu\rho\Gamma(\Qw)}.
\end{equation}
Here, $\Gamma(\cdot)$ and $\Gamma(\cdot, \cdot)$ are the complete and
incomplete Gamma functions~\cite{Abramowitz72}, respectively.

In order to prove that Poissonian arrivals lead to concave energy consumption
functions, simply substitute~\eqref{eq:burst-h-poisson}
into~\eqref{eq:convex-condition} to obtain after some straightforward
calculations the inequality
\begin{equation}
  \begin{split}
    \mu\rho \Gamma(\Qw)^2 \textrm{e}^{\mu\rho \Ts} 
    \Bigl(
      \Ts (\mu\rho \Ts)^{\Qw} \bigl( 
        (\Qw - \mu\rho
        \Ts) \Gamma(\Qw, \Ts) + \\
        \mu\rho (\Ts + \Tw) \Gamma(\Qw) 
      \bigr) 
      +\\
      2\textrm{e}^{\mu\rho \Ts}
      \Gamma(\Qw + 1, \mu\rho \Ts) \bigl( 
        (\Ts + \Tw) \Gamma(\Qw) -\\
        \Ts \Gamma(\Qw, \mu\rho \Ts)
      \bigr) 
    \Bigr) > 0.
  \end{split} 
\end{equation}
All the constant terms appearing in the above inequality are positive, so
this simplifies somewhat to
\begin{equation}
  \label{eq:churro}
  \begin{split}
  \Ts (\mu\rho \Ts)^{\Qw} \Bigl( \mu\rho \bigl(
  (\Ts + \Tw) \Gamma(\Qw) -
  \Ts \Gamma(\Qw,\mu\rho \Ts) \bigr) + \\
  \Gamma(\Qw + 1, \mu\rho \Ts)
  \Bigr) + \\
  \Gamma(\Qw + 1, \Ts \mu \rho) \bigl( (\Ts +
  \Tw) \Gamma(\Qw) - 
  \Ts \Gamma(\Qw, T_s \mu\rho)
  \bigr) > 0
\end{split}
\end{equation}
which holds true because
\begin{equation}
  \mu\rho (\Ts + \Tw) \Gamma(\Qw) > \mu \rho \Ts
  \Gamma(\Qw) > \mu\rho \Ts \Gamma(\Qw, \mu\rho \Ts)
\end{equation}
as a consequence of elementary properties of the Gamma functions. This implies
that all the summands in the left side of~\eqref{eq:churro} are positive,
and~\eqref{eq:convex-condition} is satisfied.

The last step is to prove concavity for the general
approximation~\eqref{eq:burst-h-approx}. A change of variable $m = \Qw / \rho$
transforms~\eqref{eq:burst-h-approx} into~\eqref{eq:frame-t-approx} formally.
Since $m > 0$, following the same steps as in frame transmission, one
concludes that~\eqref{eq:burst-h-approx} also fulfills
condition~\eqref{eq:convex-condition}. Hence, the link energy consumption
function $E(\cdot)$ is concave with burst transmission in the high-load
regime, regardless the traffic arrival pattern.

A numerical illustration of the concavity is shown in
\added{Fig.}~\ref{fig:contour}, which depicts the contour plots of
$h^{\prime\prime}(\rho)$ for a $10$ Gb/s Ethernet link as the traffic load and
the packet size vary.

\section{Delay Control}
\label{sec:delay-control}

According to the previous sections, a straightforward application of a
water-filling algorithm to share traffic among the bundle links provides
maximum energy savings. However, if proper care is not taken, packet delay can
grow uncontrolled as we explain next.

From a practical point of view there are many ways to implement a water
filling algorithm. For instance, one could use separate queues for each link
and only divert traffic to new links when the queue of the previous one
overflows. Obviously, this approach exhibits the greatest delay. A second
option is to limit the load factor in every link, and thus the delay, and
divert traffic when this \added{threshold} is reached. Its main drawback is
that no link is used at its full capacity and so the energy savings are not
maximum. Another option, in the opposite extreme, is to have a common bundle
queue and zero-length queues at the links. In this case, a new link is used if
when a packet arrives, the previous link is busy transmitting a packet. The
problem is that if the traffic load is not high enough, we will find that the
first link is idle while the second one is transmitting, and that goes against
the idea of the water-filling algorithm.

We propose a simple dynamic water-filling algorithm that can control average
delay, while keeping the utilization factor of the links close to 1. The
algorithm has one configuration parameter, the expected delay ($d_e$) and
$N+1$ state variables, with $N$ the number of links in the bundle, as it just
keeps a record of the short term average delay ($d_{av}$), calculated with an
exponentially weighted moving average, and the current queue length in each
link measured in time units ($q_i,\,i = 1\ldots N$). \added{More precisely, when a
new packet is about to get queued in queue $i$, the current average delay value
$d_{av}^{\prime}$ is updated as
\begin{equation}
  \label{eq:d-av-exp}
  d_{av}^{\prime} = \beta q_i + (1-\beta) d_{av}, \, 0 < \beta < 1,
\end{equation}
where $q_i^j$ is the $q_i$ value when the $j$-th packet arrives, calculated as the
amount of traffic stored in the $i$-th queue over the link capacity, and
$\beta$ is a gain factor. Updating $d_{av}$ on packet arrivals avoids the need
to record and store the arrival time of every packet to the system.}

The algorithm works as follows. Each link in the bundle is assumed to have its
own queue, so whenever a new packet arrives, the algorithm decides which queue
should store it. For this the expected delay is compared with the current
average delay. If $d_{av} < d_e$, the packet is stored in the queue of the
first link. For every other case, a sequential search is started for a queue
with a queue length smaller than the expected delay. If no queue is found, the
packet is stored in the last queue. This is all summarized in
Listing~\ref{lst:dynamic-algo}.
\begin{lstlisting}[float,label=lst:dynamic-algo,caption=Dynamic water-filling
algorithm.]
function packet_arrival
     if ($d_{av}$ < $d_e$)
	 return enqueue(Links[$1$])

     foreach ($\ell$ in Links)
         if ($q_{\ell}$ < $d_e$)
	     return enqueue($\ell$)

     return enqueue(N)
\end{lstlisting}

\section{Results}
\label{sec:results}

We have carried out several experiments to assess the effectiveness of our
proposed sharing strategy. We have employed the ns-2 network simulator with an
added module for simulating IEEE~802.3az links available for download
at~\cite{14:_networ_simul}. The simulated bundles have a varying number of
10$\,$Gb$/$s links with 10GBASE-T interfaces, so $\Ts=2.88\,\mu$s,
$\Tw=4.48\,\mu$s and $\sigma_{\mathrm{off}}=0.1$, in accordance with several
estimates provided by different manufactures during the standardization
process of the IEEE~802.3az standard. For the \emph{burst transmission}
simulations we set up $\Tmax=100\,\mu$s and $\Qw=20\,$frames, so that
$\mu\Ts>3.6\,$frames, as recommended
in~\cite{herreria12:_gi_g_model_gb_energ_effic_ether}.

\subsection{Model Validation}
\label{sec:model-validation}

The first set of experiments tests all possible traffic sharing alternatives
in a simple 2-link bundle when it is fed with synthetic traffic. For the
experiments we used a fixed frame size of 1000$\,$bytes and a varying arrival
rate, so that the aggregated load ranged between 25 and 175\%. Then, for each
load we modified the share between the two links and, for each share, we run
five simulations with different random seeds and a ten seconds duration.

\begin{figure*}
  \centering
  \subfigure[res-2-link-poisson-frame][Frame
  transmission]{\label{fig:res-2-link-poisson-frame}\includegraphics[width=\columnwidth]{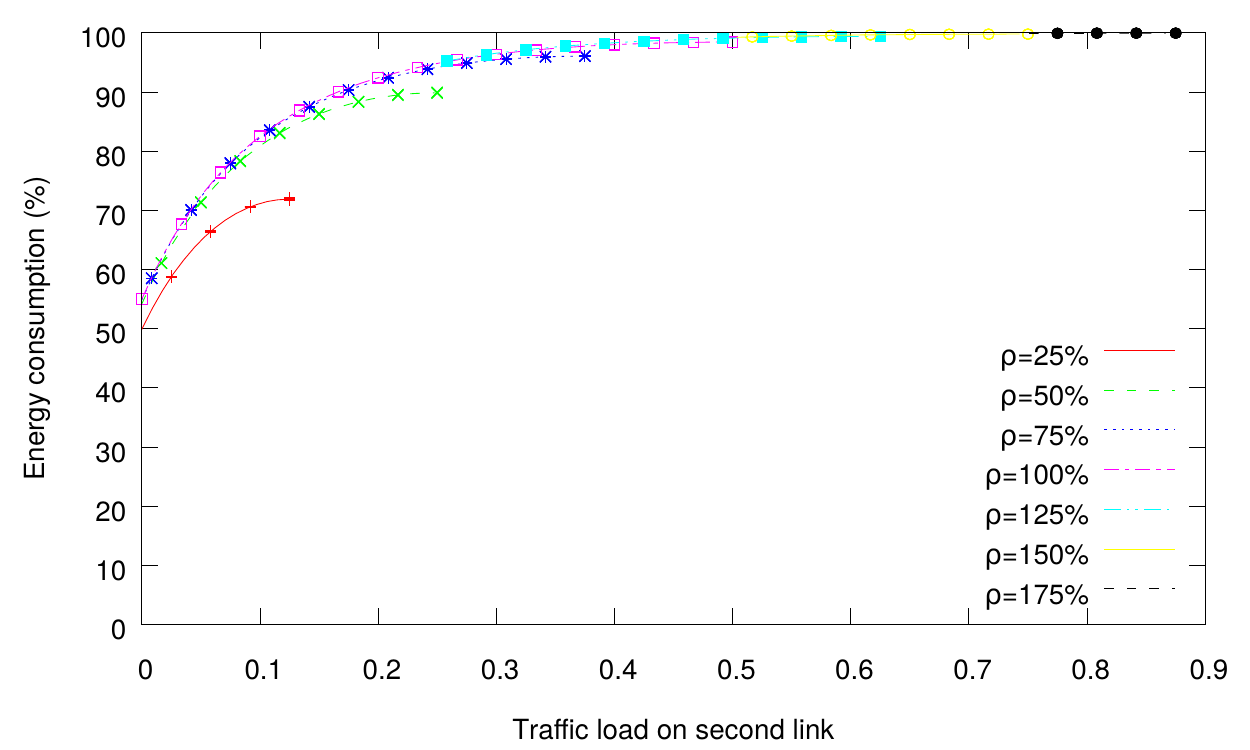}}
  \subfigure[res-2-link-poisson-burst][Burst
  transmission]{\label{fig:res-2-link-poisson-burst}\includegraphics[width=\columnwidth]{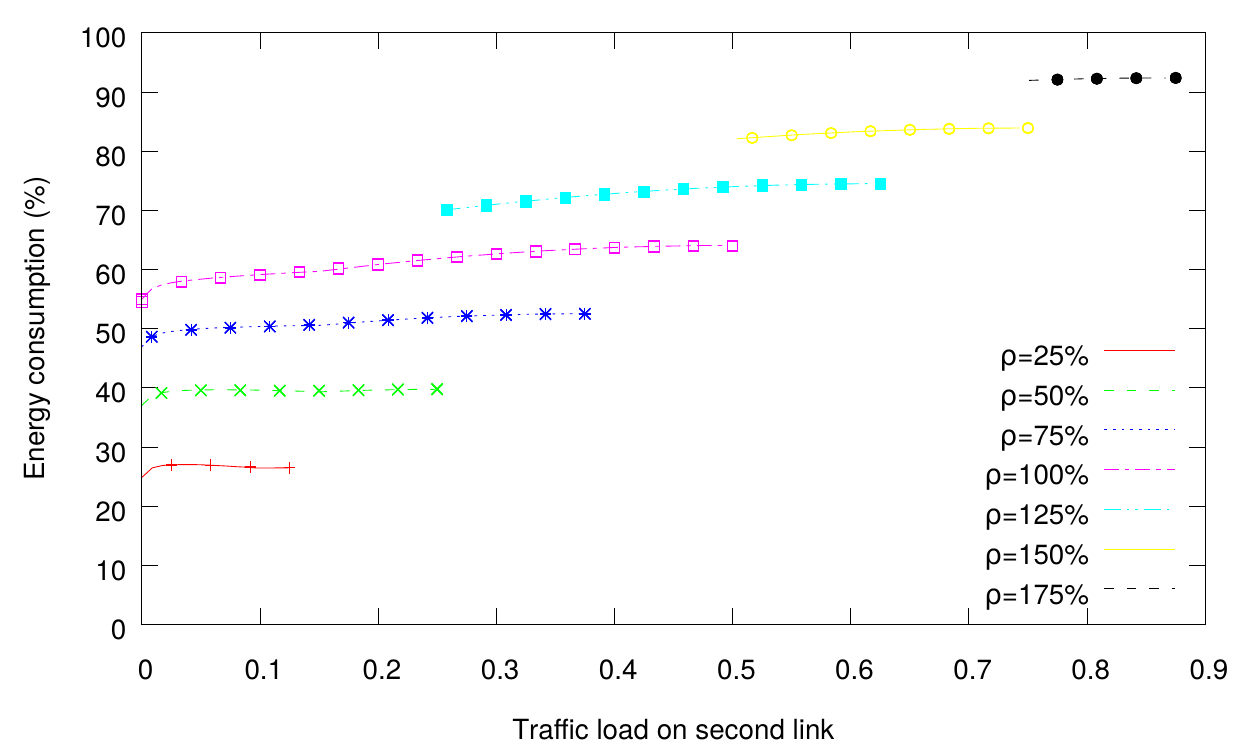}}
  \caption{Results for a 2-link bundle with Poisson traffic as a function of
    excess traffic load on the second link.}
  \label{fig:res-2-link-poisson}
\end{figure*}
Figure~\ref{fig:res-2-link-poisson} shows the total energy consumption of the
bundle versus the traffic load on the second link for Poisson traffic with
both the frame and burst transmission algorithms. For clarity, we take
advantage of the symmetry of the problem and only represent the results where
load on the second link is smaller than that on the first. \added{Thus, for
  each experiment the leftmost value represents the water-fill algorithm, with
  most of the traffic on the first link, while the rightmost value corresponds
  with an equal share of traffic among both links.}
Figure~\ref{fig:res-2-link-poisson} shows very clearly that there is very
little variance among the different simulations for the same share and load
and, at the same time, that the results match those provided by the model,
plotted with continuous lines in the graph. It is also easy to see the
increasing energy consumption with the traffic load on the second link. The
closer the loads of both links are, the higher the energy needs. In fact, the
minimum consumption is obtained when most load is concentrated on a single
link, as predicted. Finally, we also observe that the benefit of aggregating
load on a single link is much greater for frame than for burst transmission.
This is a consequence of the fact that the energy profile of the burst
transmission algorithm is more
linear\added{~\cite{christensen10:_the_road_to_eee}. Thus, there is less
  sensitivity to how the traffic load is shared among the links of the bundle
  ---note that total energy consumption shows little variations when the
  traffic share is modified in Fig.~\ref{fig:res-2-link-poisson-burst}---.}
Also, as expected, burst transmission needs less energy than frame
transmission.

\begin{figure*}
  \centering
  \subfigure[res-2-link-pareto-frame][Frame
  transmission]{\includegraphics[width=\columnwidth]{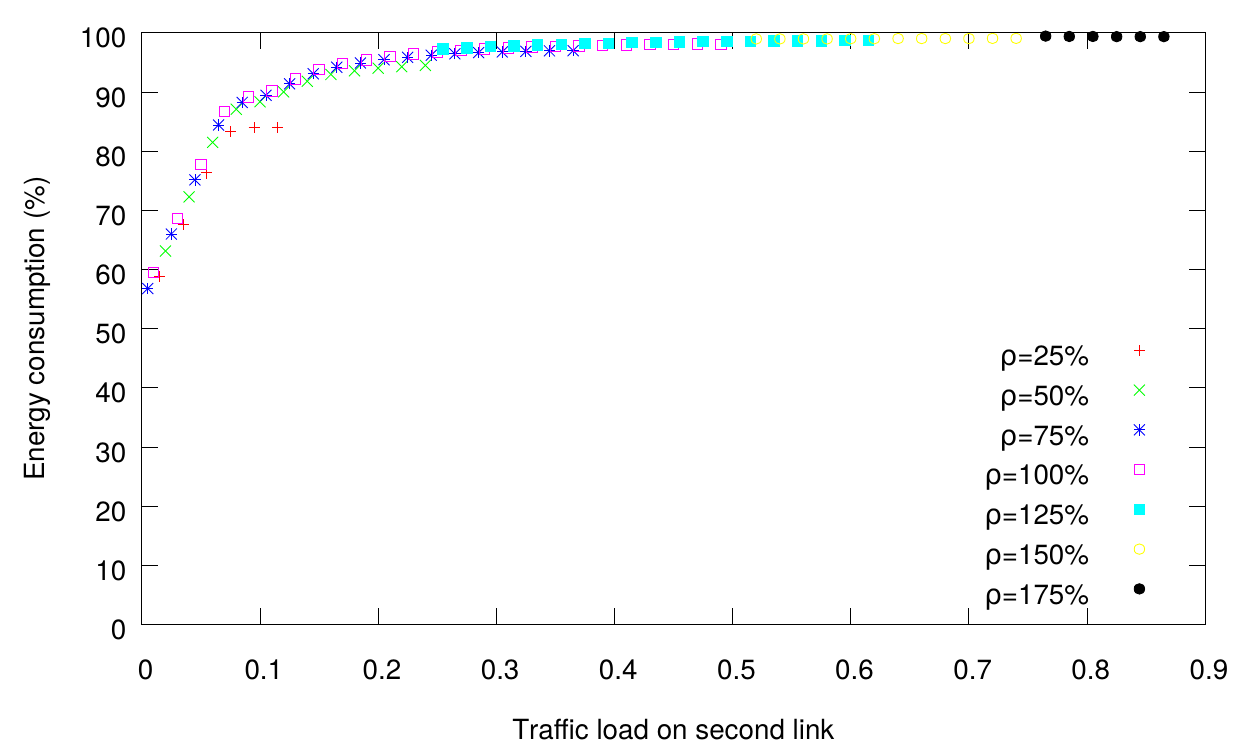}}
  \subfigure[res-2-link-pareto-burst][Burst
  transmission]{\includegraphics[width=\columnwidth]{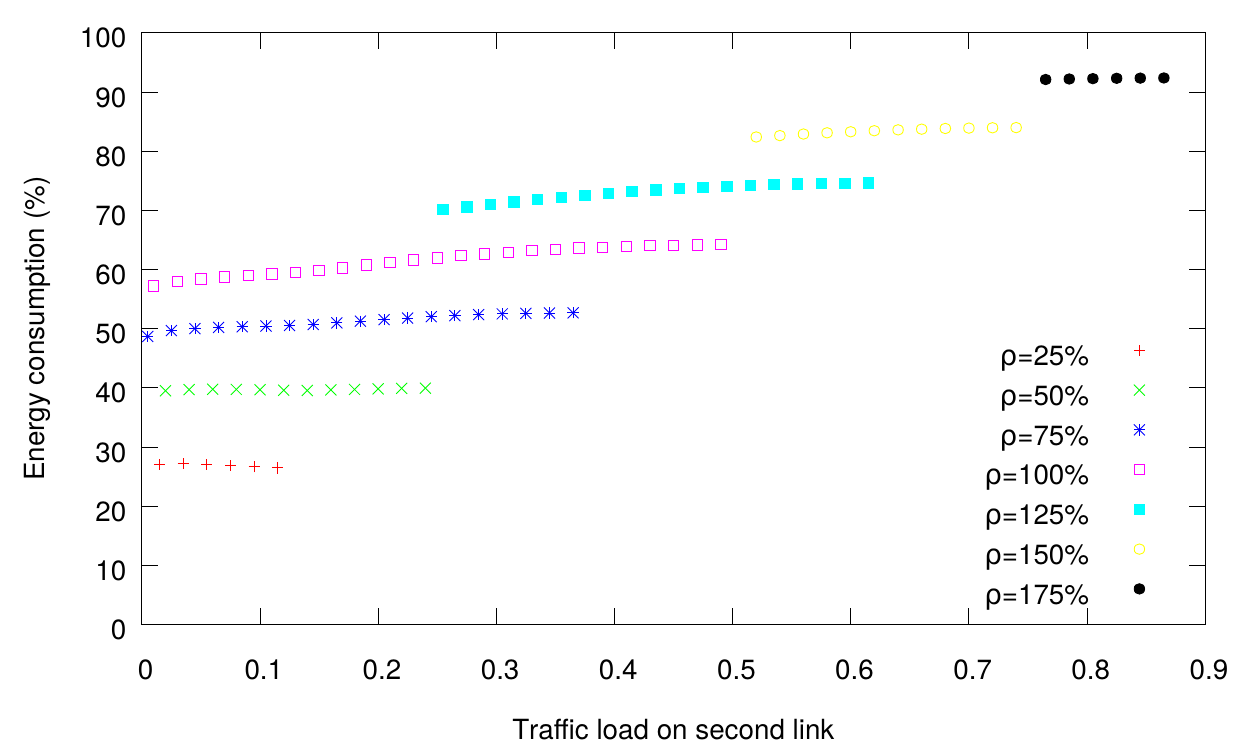}}
  \caption{\addedfragile
    {Results} for a 2-link bundle with Pareto traffic as a
    function of excess traffic load on the second link. Theoretical values
    omitted, as there is no closed form formula for Pareto
    arrivals~\cite{herreria12:_gi_g_model_gb_energ_effic_ether}.}
  \label{fig:res-2-link-pareto}
\end{figure*}
The results for Pareto traffic (with the shape factor $\alpha$ set to
$2.5$)\footnote{Pareto distributions must be characterized with a shape
  parameter $\alpha$ greater than 2 to have a finite variance. However, the
  greater the $\alpha$ parameter is, the shorter the fluctuations, so a value
  of $2.5$ is a good compromise to have finite variance along with significant
  fluctuations.} are plotted in \added{Fig.}~\ref{fig:res-2-link-pareto}. Although
the performance curves are not as smooth as for the Poisson traffic, the
previous conclusions still hold. Again, the minimum consumption is obtained
when most of the traffic is on a single link and then increases as the traffic
on the second link increases. At the same time, the frame transmission
algorithm benefits more than the burst transmission one.

Our second experiment compares the overall energy consumption of an Ethernet
bundle for the full range of possible incoming traffic demands and two
different sharing methods. The first spreads the traffic evenly across all the
constituent links, denoted in the results by \emph{equitable,} while the
second is the naïve water-filling method. Traffic follows a Poisson
distribution and the frame size is 1000$\,$bytes, as in the previous
experiment. Figure~\ref{fig:limits-exp} displays both the experimental and
analytic results for two, four and eight-link aggregates. Again, frame
transmission algorithm benefits more than burst transmission of the water-fill
sharing algorithm. Further, as the number of links in the bundle increases,
the energy demands of frame transmission, when using the water-fill procedure,
approximate those of burst transmission.

\subsection{Dynamic Water-filling Algorithm}
\label{sec:dynam-water-fill}

The next set of experiments tests the behavior of the dynamic water-filling
algorithm. We have employed real traffic traces captured on Internet backbones
for the simulations. The traffic comes from the publicly available passive
monitoring CAIDA dataset from
2013~\cite{internet13:_caida_ucsd_anony_inter_traces} which provides
anonymized traces from a 10$\,$Gb$/$s Internet backbone. We used one of these
traces to feed traffic to a simulated 4-link bundle made of 10GBASE-T
interfaces. Of all the available traces, we have chosen one with a relatively
high demand of about 6$\,$Gb$/$s on average. As that load is \added{still}
quite low for our simulated bundle of 40$\,$Gb$/$s we made new traces of
approximately 12, 18, 24 and 30$\,$Gb$/$s combining traffic from additional
independent adjacent traces\added{. For this we concatenated the traces and
  then reduced the inter-arrival times by a constant factor (2, 3, 4 and 5
  respectively). We proceeded in this manner to keep any existing
  auto-correlation in the final traces.} \added{Finally, we have chosen
  $\beta=0.1$ as the gain factor in~\eqref{eq:d-av-exp}.}
\begin{figure}
  \centering
  \subfigure[Two links bundle]{
    \includegraphics[width=\columnwidth]{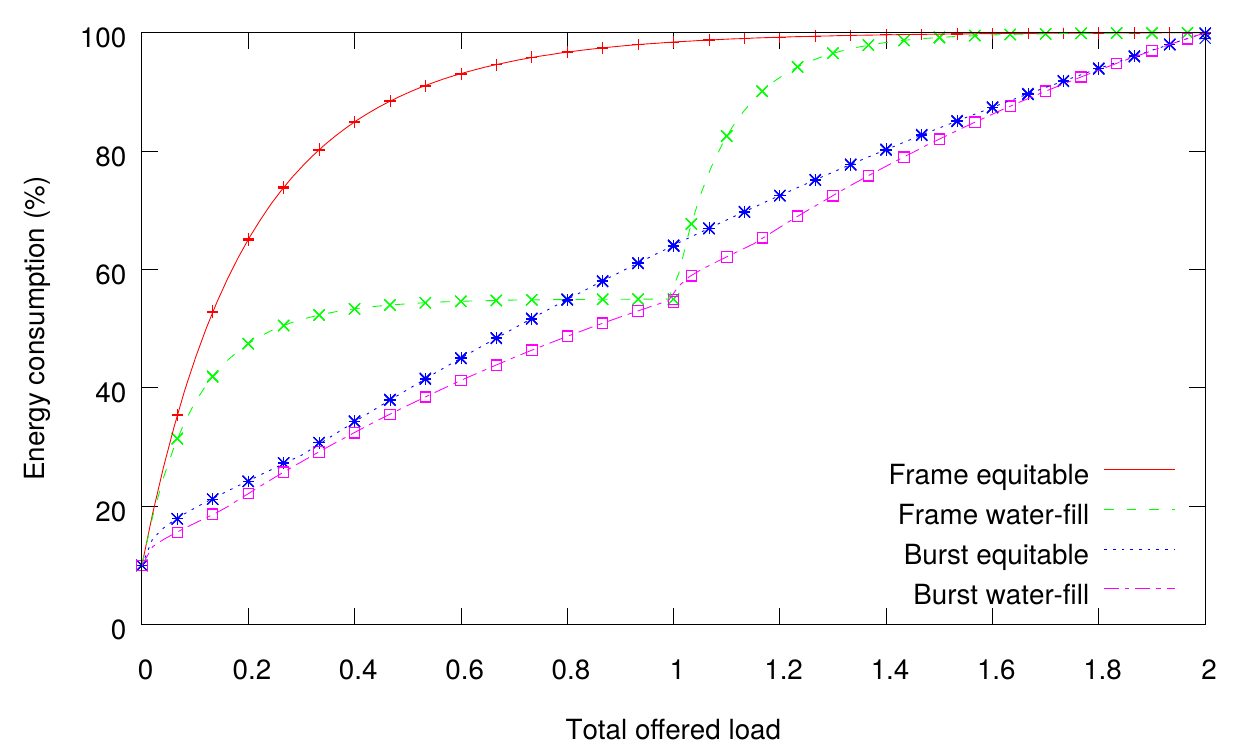}}
  \subfigure[Four links bundle]{
    \includegraphics[width=\columnwidth]{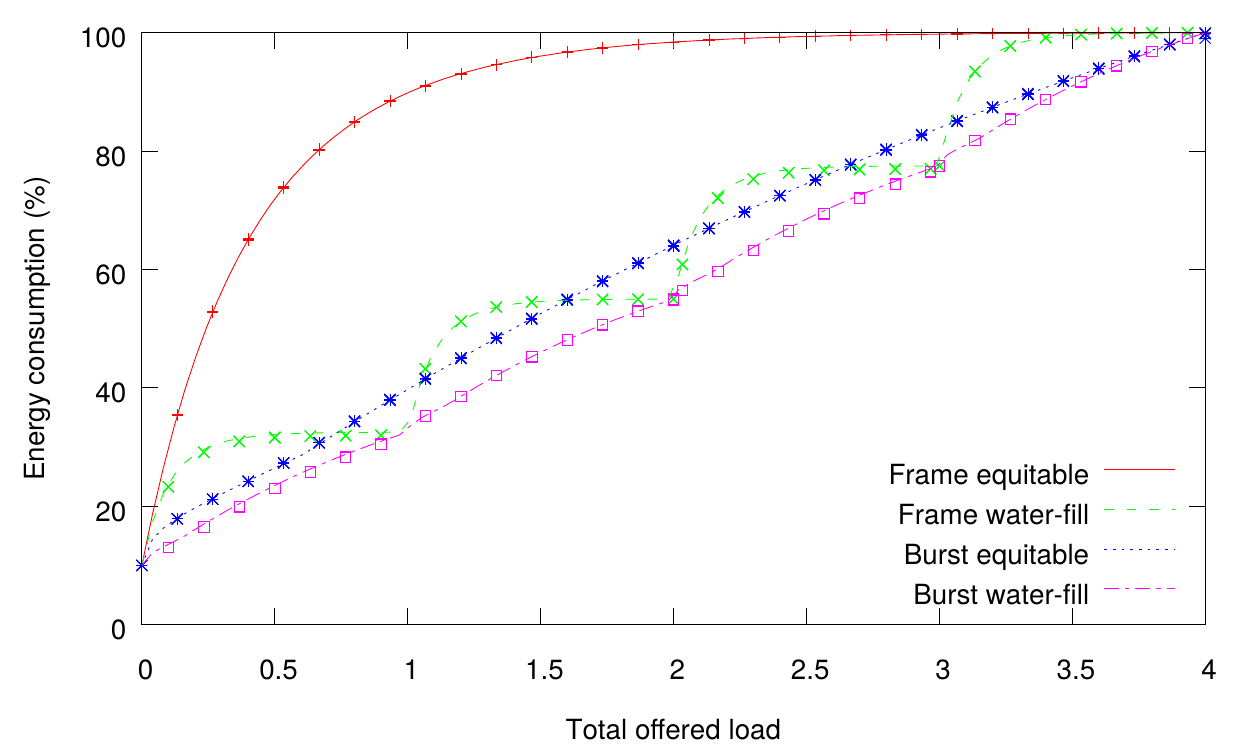}}
  \subfigure[Eight links bundle]{
    \includegraphics[width=\columnwidth]{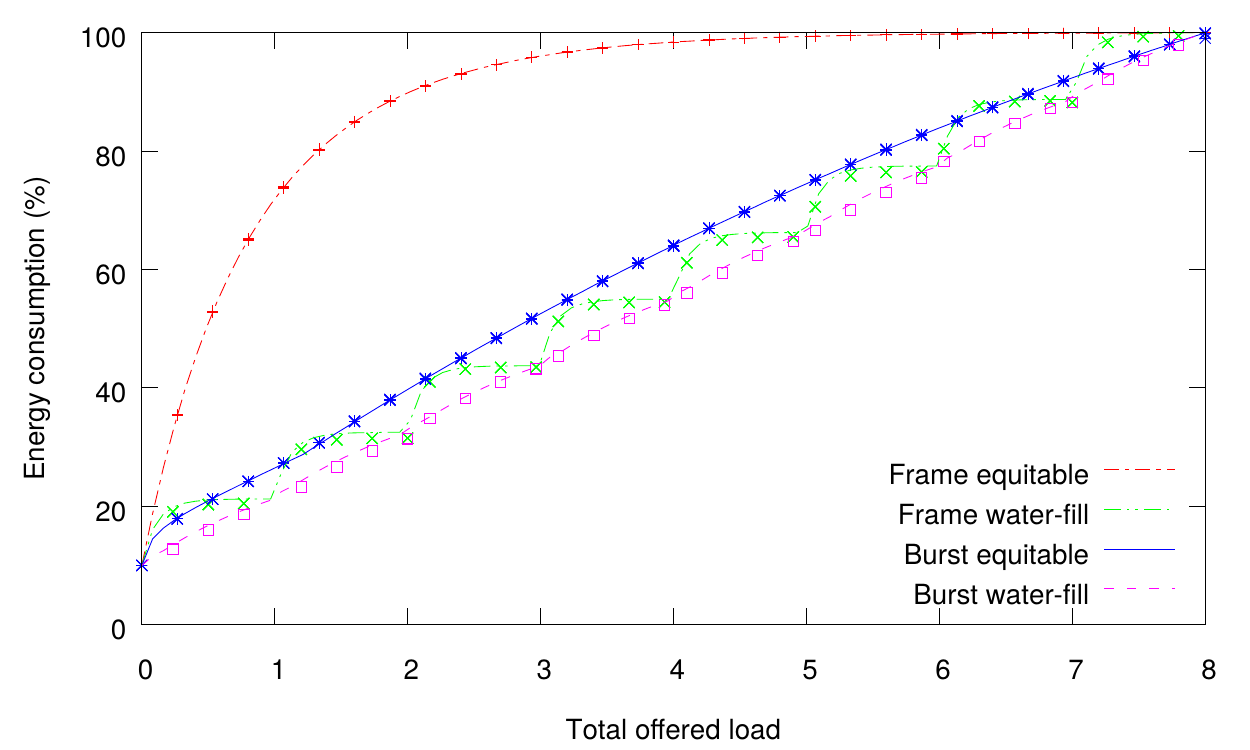}}
  \caption{\label{fig:limits-exp} Normalized global consumption of a bundle
    link for the different idle mode governors.}
\end{figure}

The first experiment verifies that the \added{proposed dynamic }algorithm is
in fact able to control the average delay. For this we have fed all the
traffic traces to a 4-link bundle, and configured the algorithm for different
expected delays. The results are plotted in
\added{Fig.}~\ref{fig:delayvsdelay}.\footnote{Results for burst transmission
  have been omitted for the sake of brevity, but show a similar behavior.} It
can be clearly seen an almost perfect relationship between the configured and
the measured average delay for values greater than the transition times of the
EEE links. The exception is the 6$\,$Gb$/$s trace, that is bounded below
4$\,\mu$s. This is expected, as the queue cannot grow larger when the capacity
of a single link is greater than the offered traffic. The simulation with the
12$\,$Gb$/$s trace shows a small drift of the average delay, but, in any case,
the average delay is kept below the configured delay. \added{This error in the
  12$\,$Gb$/$s experiment occurs because~\eqref{eq:d-av-exp} can overestimate
  average queuing delay if waiting time samples from low used queues are few
  and far between them, so the samples from the first queue get
  over-represented. Although omitted for brevity, decreasing the $\beta$ value
  lessens the drift.}
\begin{figure}
  \centering
  \includegraphics[width=.5\textwidth]{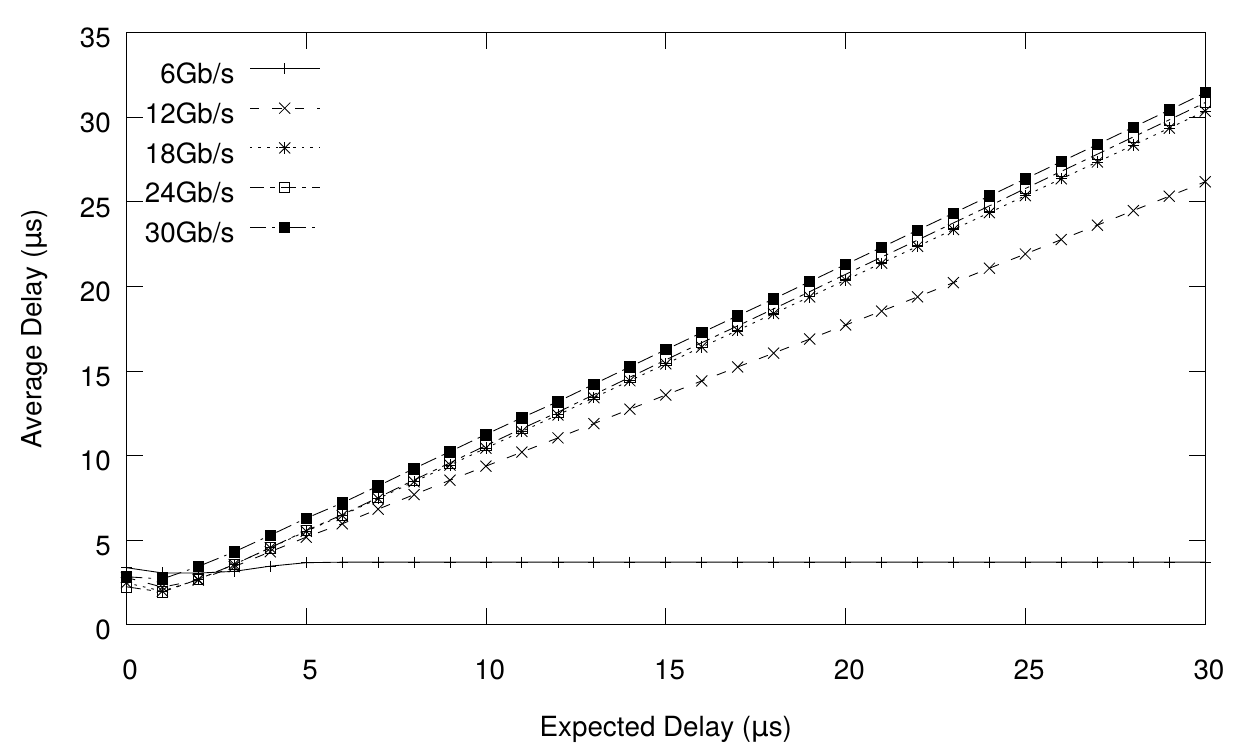}
  \caption{Obtained average delay versus configured delay for the dynamic water-fill algorithm.}
  \label{fig:delayvsdelay}
\end{figure}

The second experiment shows the variation of power consumption versus expected
delay. The results are shown in \added{Fig.}~\ref{fig:powervsdelay}. When the
expected delay is too low, all links are used simultaneously, and the power
savings are minimal. However, as the allowed delay increases, most of the
traffic is transmitted by the first links and, despite the fact that all of
them are powered on, we achieve large power savings thanks to the concavity of
the cost function. It is important to notice that the maximum energy savings
are already obtained starting from low delay target values. This allows to
deploy the algorithm even in networks used by delay-sensitive applications.
\begin{figure}
  \centering
  \includegraphics[width=.5\textwidth]{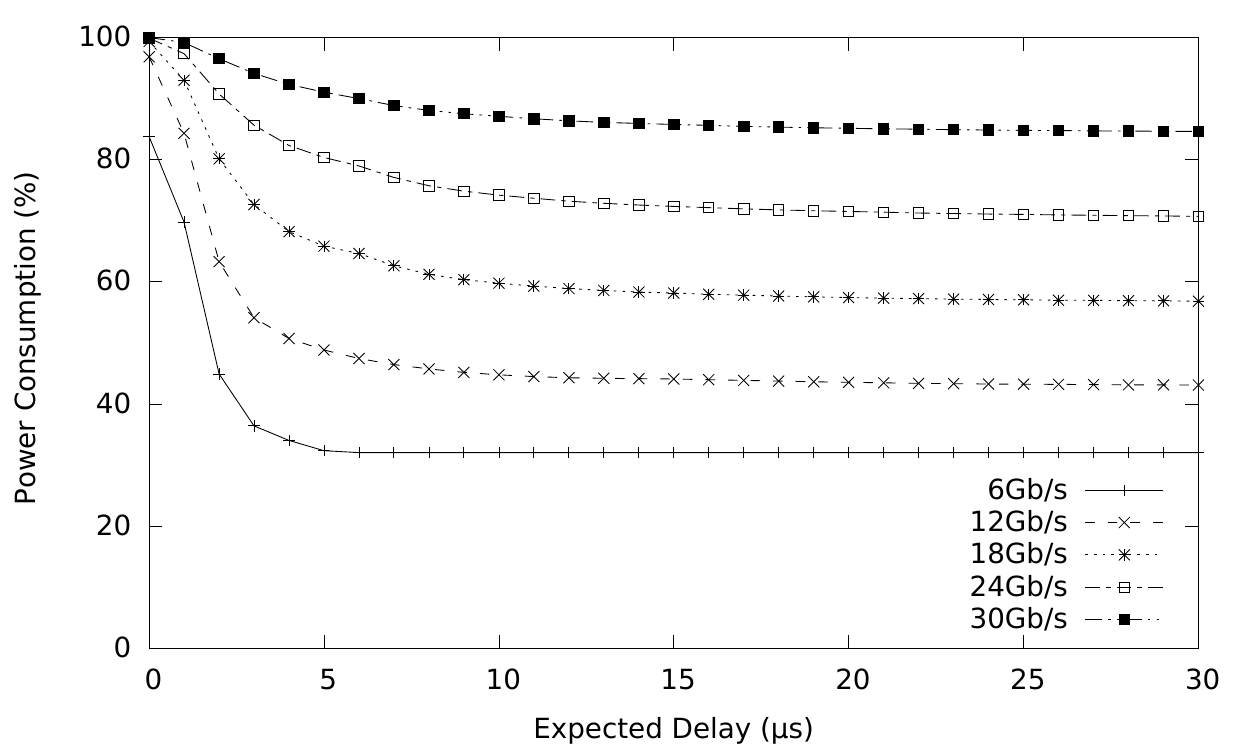}
  \caption{Power consumption versus expected delay for the dynamic water-fill algorithm.}
  \label{fig:powervsdelay}
\end{figure}

In the last experiment we have compared the results obtained when sharing the
traffic with three different strategies: spreading the traffic evenly across
the four links, that we called \emph{equitable}, a naïve implementation of the
\emph{water-fill} algorithm and, finally, the \emph{dynamic} water-fill
algorithm with a target delay of ten microseconds for the frame transmission
algorithm and 20$\,\mu$s for the burst one.\footnote{In burst transmission
  power savings reach their maximum for a higher delay value than frame
  transmission. This is expected as burst transmission adds additional delay
  in the form of queuing before waking up a link.} For the naïve
implementation we have constrained the traffic load on any link to 90\% to
avoid excessive buffering.

\begin{table}
  \centering
    \begin{tabular}{rlcccc}
      \multicolumn{1}{c}{Bundle} &
      \multicolumn{1}{c}{Strategy} & 
      \multicolumn{1}{c}{Link \#1} & 
      \multicolumn{1}{c}{Link \#2} & 
      \multicolumn{1}{c}{Link \#3} & 
      \multicolumn{1}{c}{Link \#4} \\\hline
      & Equit. & $1.55$ &$1.55$ &$1.55$ &$1.55$ \\
      $6.21$ & Naïve Water-fill & $6.21$ & 0 & 0 & 0 \\
      & Dyn. Frame & $6.21$ & 0 & 0 & 0\\    
      & Dyn. Burst & $6.18$ & $0.03$ & 0 & 0\\\hline    
      & Equit. & $3.15$ &$3.15$ &$3.15$ &$3.15$ \\
      $12.60$ & Naïve Water-fill & $9$ & $3.60$ & 0 & 0 \\
      & Dyn. Frame & $9.44$ & $3.08$ & $0.07$ & 0 \\
      & Dyn. Burst & $9.50$ & $2.67$ & $0.39$ & $0.04$\\\hline    
      & Equit. & $4.71$ &$4.70$ &$4.70$ &$4.70$ \\
      $18.81$ & Naïve Water-fill & $9$ & $9$ & $0.81$ & 0 \\
      & Dyn. Frame & $9.94$ & $7.12$ & $1.73$ & $0.02$ \\
      & Dyn. Burst & $9.97$ & $6.82$ & $1.77$ & $0.25$\\\hline     
      & Equit. & $6.27$ &$6.27$ &$6.27$ &$6.27$ \\
      $25.08$ & Naïve Water-fill & $9$ & $9$ & $7.08$ & 0 \\
      & Dyn. Frame & $10$ & $9.16$ & $5.12$ & $0.80$ \\
      & Dyn. Burst & $10$ & $9.13$ & $4.78$ & $1.17$\\\hline     
      & Equit. & $7.85$ &$7.85$ &$7.85$ &$7.85$ \\
      $31.40$ & Naïve Water-fill & $9$ & $9$ & $9$ & $4.4$ \\
      & Dyn. Frame & $10$ & $9.82$ & $7.94$ & $3.64$ \\
      & Dyn. Burst & $10$ & $9.83$ & $7.74$ & $3.83$\\\hline     
    \end{tabular}
  \caption{Average traffic fed into each link for the real traffic simulations
    (in Gb$/$s).}
  \label{tab:caida}
\end{table}
The exact traffic rate of each trace and the different shares are detailed in
Table~\ref{tab:caida}. For the equitable and the naïve water-fill they have
been determined beforehand, but for the dynamic algorithm the table lists the
results obtained via simulation. The results for both the frame transmission
and the burst transmission algorithms are depicted in
\added{Fig.}~\ref{fig:caida}.
\begin{figure}
  \centering
  \includegraphics[width=.5\textwidth]{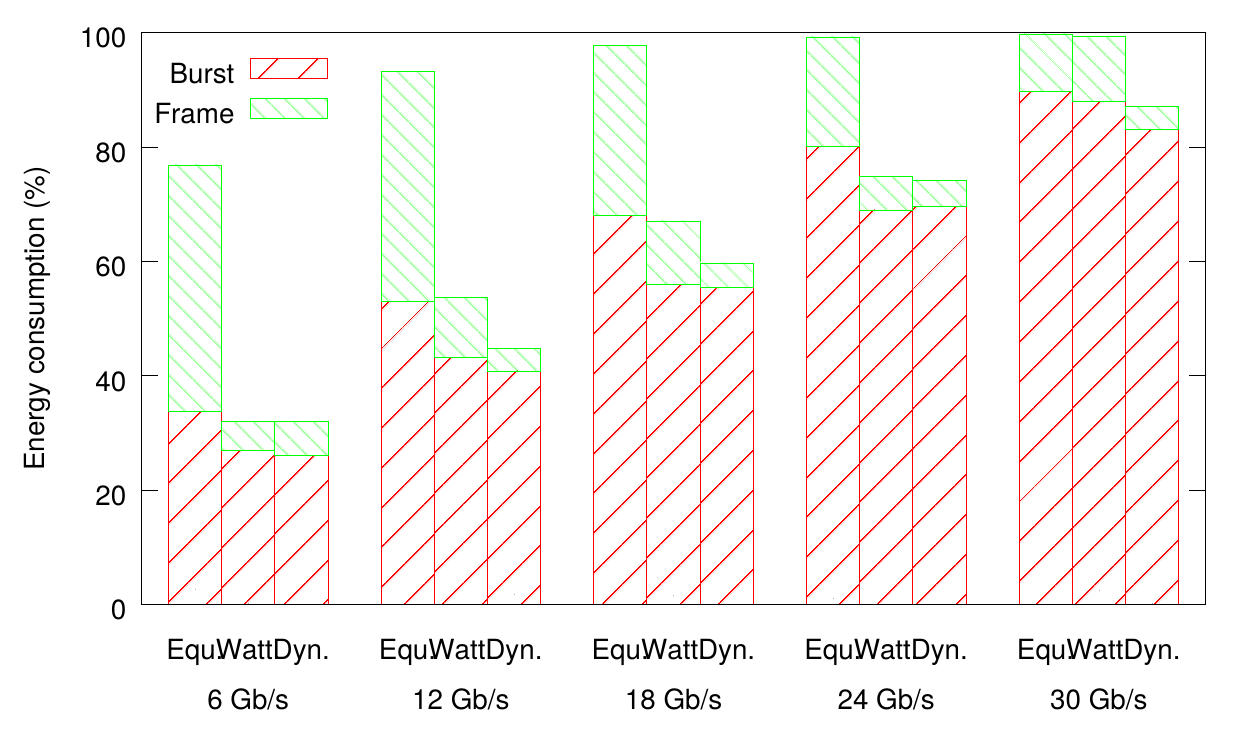}
  \caption{Energy consumption with real traffic traces when employing
    different strategies to share the traffic in a 4-link bundle.}
  \label{fig:caida}
\end{figure}
In every case the frame transmission algorithm needs more energy than the
burst transmission one, but, at the same time, the savings resulting from
applying the water-fill procedure are also greater. In fact, there is usually
very little difference in the consumption of both EEE algorithms in that case.
As expected, the equitable share draws more energy than the other two shares
and the water-fill share is the one that produces the best results. Finally,
the dynamic water-fill algorithm improves the results, but not substantially.
 
We have also measured the impact of the different algorithms on queuing delay.
Figure~\ref{fig:caida-delay} shows the average queuing delay suffered by the
traffic in the previous experiment.
\begin{figure}
  \centering
  \includegraphics[width=.5\textwidth]{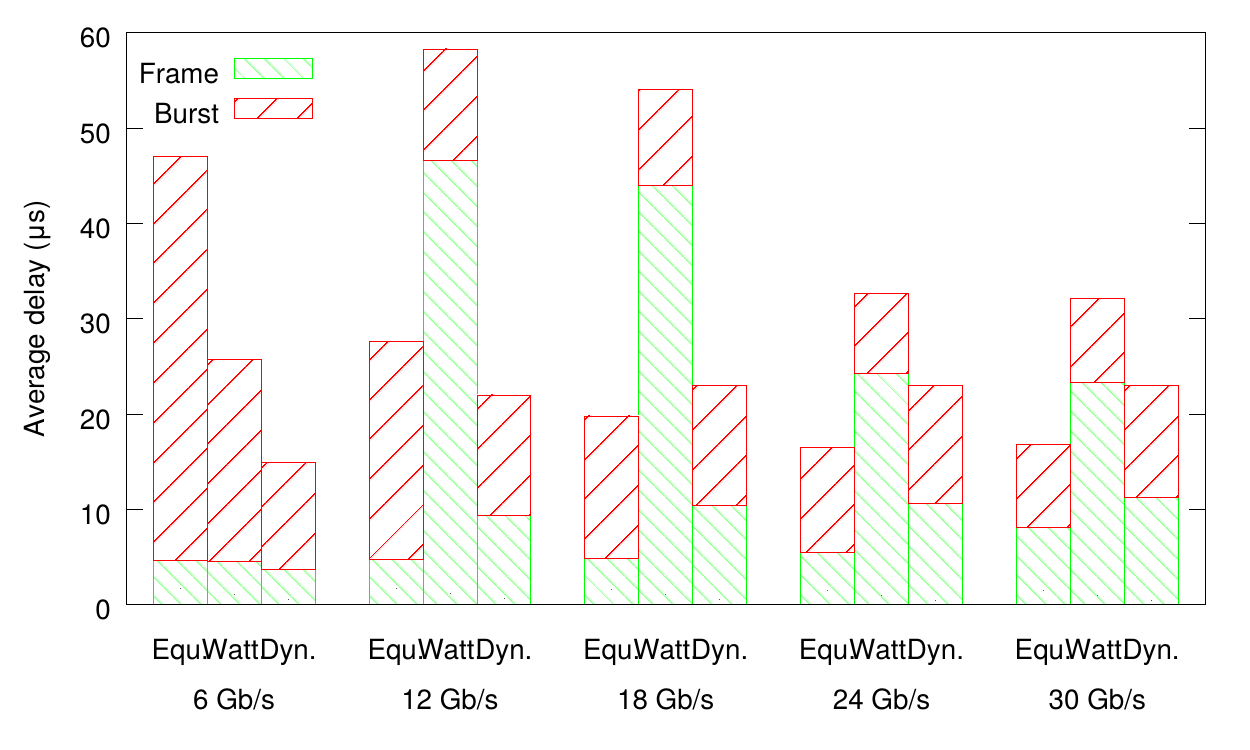}
  \caption{Average queuing delay with real traffic traces when employing
    different strategies to share the traffic in a 4-link bundle.}
  \label{fig:caida-delay}
\end{figure}
As it is the case for single EEE
links~\cite{herreria12:_gi_g_model_gb_energ_effic_ether}, we observe that
burst transmission always causes more delay than frame transmission. We also
find that the different sharing methods impact on the queuing delay
differently. \added{In the 6$\,$Gb$/$s case, the equitable share produces the
  highest delays with burst transmission, as there is relatively little
  traffic in every link, and thus packets have to wait for $Q_w$ packets to
  arrive before being transmitted. In every other case, the naïve water-fill
  algorithm produces the longest packet delays, as at least one queue is
  driven near its full capacity. Finally, the dynamic algorithm produces
  stable delays, near its target ---10$\,\mu$s for frame and $20\,\mu$s for
  burst transmission--- that are in the same range as those of the equitable
  share.}

\section{Conclusions}
\label{sec:conclusions}

This paper presents an optimum, yet simple, procedure for distributing traffic
load among the links of a bundle that minimizes energy consumption when
individual links employ an EEE algorithm. As explained, the maximum energy
savings are obtained when traffic is only transmitted on a link if all the
previous ones in the aggregate are already being used at their maximum allowed
load. The paper proves the optimality of the procedure for typical energy cost
functions of individual Ethernet links.

The provided procedure is oblivious of the energy saving algorithm used in the
links, whether it is the simple \emph{frame transmission} algorithm or the
more efficient \emph{burst transmission} one. Moreover, we found that as the
number of links forming the bundle increases, the difference in the total
energy consumption between both algorithms vanishes when using our sharing
procedure. Thus, for bundles made up of many links it is advisable to use the
simpler \emph{frame transmission} algorithm in the links, as it both reduces
complexity and adds less latency to the transmitted frames.

We have also explored several alternatives to build a practical implementation
of the water-filling idea to then present a simple practical implementation
that is able to keep average delay controlled at a configurable target value
while minimizing overall energy consumption. \added{The algorithm requires
  little memory and computational power, so that a vendor can implement it
  just by modifying the firmware of the Ethernet line card. However, as the
  algorithm needs to obtain the queue occupation of each port to classify
  incoming packets, an open-flow implementation is not currently possible, as
  the current spec~\cite{foundation13:_openf_switc_specif} does not define the
  needed counters. Future research could explore the possibility of extending
  the current spec to empower the user with the fine grained control of the
  transmission ports needed by our proposal.}

Finally, we have tested our procedure with both synthetic and real
traffic traces. In all cases, the obtained results match our expectations with
the best results being obtained when the proposed sharing algorithm is
employed, reducing energy consumption as much as $50\,$\%.

\appendices

\section{Proof of proposition~\ref{propo:waterfill}}
\label{sec:proof}

In this Section, we prove that for the particular case of equal cost functions
the solution to the optimum allocation is a simple sequential water-filling
algorithm: each link capacity is fully used before sending traffic through a
new, idle link. 

We assume $X < \sum_i C_i$, otherwise the solution is trivial. It is easy to
see that the constraints define a convex region $\mathcal{R}$. Since the
objective function is concave, it follows that it attains its minimum at some
of the extreme points of $\mathcal{R}$, namely $x_i=C_i$ for $i\in T \subseteq
[1:N],\quad 0 < x_j < C_j$ for one $j \in [1:N]$ and $x_k=0$ for all $k \in
T^C\setminus \{j\}$. In fact, when all the cost functions are equal, the
optimal traffic allocation is to use the links in decreasing order of
capacity. Assume, without loss of generality, that
$C_1>C_2>\ldots>C_N$.\footnote{If some links are of the same capacity, each
  permutation of the links lead to an equivalent solution of the problem.} Fix
two links $i$ and $j$, $i>j$, and assume that a feasible solution is the
vector $\mathbf{x}=(x_i^*,\ldots,x_N^*)$. Then, since $E$ is a concave
function it is also subadditive, and for $i>j$ and $\delta < \min\{x_i^*,
C_j-x_j^*\}$ we have
\begin{equation}
  \label{eq:proof-subadditive}
  E(x_i^*) + E(x_j^*) \ge E(x_i^*-\delta) + E(x_j^* + \delta).
\end{equation}

Therefore, the vector $\tilde{\mathbf{x}} = (x_1^*, \ldots, x_j^*+\delta,
\ldots, x_i^*-\delta, \ldots, x_N^*)$ is a better solution tha\added{n} $\mathbf{x}$.
Iterating this argument as many times as necessary, it is immediate to
conclude that
\begin{align}
  \label{eq:proof-results}
  x_i^*=C_i  \text{ for } i=1, \ldots, s-1\\
  0 \le x_s^*=X-\sum_{i=1}^{s-1}C_i<C_s\\
  x_j^* = 0 \text{ for } j=s+1, \ldots, N
\end{align}
is the optimal solution, where $\sum_{i=1}^{s-1}C_i \le X <
\sum_{i=1}^sC_i.\qed$

To see~\eqref{eq:proof-subadditive}, recall that for a concave function $f$
and three ordered points $a<b<c$ it holds
\begin{equation}
  \label{eq:proof-sub-first}
  \frac{f(b)-f(a)}{b-a} \ge \frac{f(c)-f(a)}{c-a}.
\end{equation}
Just let $t=(b-a)/(c-a)$ so that $b=(1-t)a+tc$.
By the definition of concavity
\begin{equation}
  \label{eq:proof-sub-second}
  f(b)\ge(1-t)f(a)+tf(c),
\end{equation}
which is~\eqref{eq:proof-sub-first}. Similarly, for $a<b<c$
\begin{equation}
  \label{eq:proof-sub-third}
  \frac{f(c)-f(a)}{c-a}\ge\frac{f(c)-f(b)}{c-b},
\end{equation}
and combining~\eqref{eq:proof-sub-first} and~\eqref{eq:proof-sub-third} gives
\begin{equation}
  \label{eq:proof-sub-last}
  \frac{f(b)-f(a)}{b-a}\ge\frac{f(c)-f(b)}{c-b}.
\end{equation}
Now, use inequality~\eqref{eq:proof-sub-last} twice over the tuples
$(x_i-\delta, x_i, x_j)$ and $(x_i, x_j, x_j+\delta)$ to
conclude~\eqref{eq:proof-subadditive}.\qed

\section{Proof of proposition~\ref{prop:convex}}
\label{sec:proof-proposition-convex}

Consider the auxiliary function
\begin{equation}
  u(x) = 1 - t(x) = a(1 - x)\frac{f(x)}{f(x) + b} = g(x) h(x),
\end{equation}
where $g(x) \triangleq a (1 - x)$ and $h(x) \triangleq f(x) / (f(x) + b)$.
Strict concavity of $t(x)$ is equivalent to $u(x)$ being strictly convex or,
alternatively, to $u^{\prime\prime}(x) > 0$. Taking the second derivative of
$u(x)$ we get
\begin{equation}
  \label{eq:convexity}
  u^{\prime\prime}(x) = g(x) h^{\prime\prime}(x) - 2 a h^\prime(x),
\end{equation}
because $g^\prime(x) = -a$. So, $u(x)$ is strictly convex if and only if $g(x)
h^{\prime\prime}(x) > 2 a h^\prime(x)$. But
\begin{equation}
  \label{eq:hprime}
  h^\prime(x) = b \frac{f^\prime(x)}{\bigl( f(x) + b \bigr)^2} < 0,
\end{equation}
since we assumed $f(x)$ to be decreasing. With $g(x) > 0$ for $x \in [0, 1]$,
$a > 0$ and~\eqref{eq:hprime}, \eqref{eq:convexity} shows that $h(x)$ convex
implies $u(x)$ convex. Finally,
\begin{equation}
  h^{\prime\prime}(x) = b \frac{f^{\prime\prime}(x) \bigl( f(x) + b \bigr) -
    2 \bigl( f^\prime(x) \bigr)^2}{\bigl( f(x) + b \bigr)^3},
\end{equation}
and $h(x)$ is a convex function ---$t(x)$ is a concave function--- if and only
if $f^{\prime\prime}(x) (f(x) + b) > 2 (f^\prime(x))^2$, since $f(x)$ is
nonnegative.\qed

\balance
\bibliographystyle{IEEEtran}
\bibliography{IEEEfull,biblio}

\end{document}